\documentclass[12pt]{article}

\usepackage{amsmath,amsfonts,amssymb,graphicx,psfrag,rotating}
\usepackage[colorlinks=true,urlcolor=blue,citecolor=black,linkcolor=blue]{hyperref}
\usepackage{setspace}
\usepackage{dcolumn}

\usepackage[superscript,biblabel]{cite}
\usepackage{comment}
\usepackage{multirow}
\usepackage{lscape}
\usepackage{bbm}
\usepackage{scalefnt}
\usepackage[autolanguage]{numprint}
\usepackage{booktabs}
\usepackage{bm}
\usepackage{threeparttable}
\usepackage{rotating}
\usepackage{color}
\usepackage{setspace}
\usepackage{sectsty}
\usepackage{ulem}
\usepackage{caption}
\usepackage{subcaption}
\usepackage{color} 
\definecolor{qz}{RGB}{255,0,0}

\usepackage{enumitem} 

\usepackage{lmodern}
\fontfamily{lmtt}\selectfont
\usepackage[T1]{fontenc}
\usepackage{float}
\usepackage{authblk}
\usepackage{xr-hyper}
\newtheorem{theorem}{Theorem}[section]

\newtheorem{corollary}{Corollary}[section]

\newenvironment{proof}[1][Proof]{\noindent\textbf{#1.} }{\ \rule{0.5em}{0.5em}}

\usepackage[top=1in, left=1in, right=1in,
bottom=1in]{geometry}

\makeatletter
\renewcommand{\section}{\@startsection{section}{1}{0em}{\baselineskip}{0.5\baselineskip}{\large\bfseries\large}}
\renewcommand{\subsection}{\@startsection{subsection}{0}{0em}{\baselineskip}{0.5\baselineskip}{\normalfont\bfseries\normalsize}}
\makeatletter

\newcommand\independent{\protect\mathpalette{\protect\independenT}{\perp}}
\def\independenT#1#2{\mathrel{\rlap{$#1#2$}\mkern2mu{#1#2}}}

\newcolumntype{.}{D{.}{.}{-1}}
\newcolumntype{d}[1]{D{.}{.}{#1}}

\begin{document}

%
\sectionfont{\bfseries\large\sffamily}%
%

\subsectionfont{\bfseries\sffamily\normalsize}%
%

\noindent
\begin{center}
\noindent\textbf{Profiling Compliers in Instrumental Variables Designs}%
\end{center}
%

\vspace{1cm}

\noindent
\textsf{Dominik Hangartner, Moritz Marbach, Leonard Henckel, Marloes H. Maathuis, Rachel R. Kelz, and Luke Keele}%
\vspace{1cm}

\noindent
\textsf{Correspondence to Luke Keele, Department of Surgery, University of Pennsylvania, 3400 Spruce St, Philadelphia, Pennsylvania, 19104 (email: luke.keele@gmail.com, phone: 614-448-8168)}%

\vspace{1cm}

\noindent\textsf{Author Affiliations: London School of Economics, London WC2A 2AE, United Kingdom (Dominik Hangartner); Center for Comparative and International Studies, ETH Zurich, 8092 Zurich, Switzerland (Dominik Hangartner and Moritz Marbach); Department of Mathematics, ETH Zurich, 8092 Zurich, Switzerland (Leonard Henckel and Marloes H. Maathuis); Department of Surgery, University of Pennsylvania, Philadelphia, Pennsylvania (Rachel Kelz and Luke Keele)}%

\vspace{1cm} 

\noindent\textsf{Sources of financial support: Dr. Keele and Dr. Kelz are funded by a grant from the National Institute on Aging, 1R01AG060612-01A1.}

\vspace{1cm}

\noindent\textsf{Conflict of interest: None.}

\vspace{1cm}

\noindent\textsf{Declarations: The dataset used for this study was purchased with a grant from the Society of American Gastrointestinal and Endoscopic Surgeons. Although the AMA Physician Masterfile data is the source of the raw physician data, the tables and tabulations were prepared by the authors and do not reflect the work of the AMA. The Pennsylvania Health Cost Containment Council (PHC4) is an independent state agency responsible for addressing the problems of escalating health costs, ensuring the quality of health care, and increasing access to health care for all citizens. While PHC4 has provided data for this study, PHC4 specifically disclaims responsibility for any analyses, interpretations or conclusions. Some of the data used to produce this publication was purchased from or provided by the New York State Department of Health (NYSDOH) Statewide Planning and Research Cooperative System (SPARCS). However, the conclusions derived, and views expressed herein are those of the author(s) and do not reflect the conclusions or views of NYSDOH. NYSDOH, its employees, officers, and agents make no representation, warranty or guarantee as to the accuracy, completeness, currency, or suitability of the information provided here.}

\vspace{1cm}

\noindent\textsf{Suggested running title: Profiling Compliers}

\clearpage
\begin{abstract}
\noindent
Instrumental variable (IV) analyses are becoming common in health services research and epidemiology. IV analyses can be used both to analyze randomized trials with noncompliance and as a form of natural experiment. In these analyses, investigators often adopt a monotonicity assumption, which implies that the relevant effect only applies to a subset of the study population known as compliers. Since the estimated effect is not the average treatment effect of the study population, it is important to compare the characteristics of compliers and non-compliers. Profiling compliers and non-compliers is necessary to understand what subpopulation the researcher is making inferences about, and an important first step in evaluating the external validity (or lack thereof) of the IV estimate for compliers. Here, we discuss the assumptions necessary for profiling, which are weaker than the assumptions necessary for identifying the local average treatment effect if the instrument is randomly assigned. We then outline a simple and general method to characterize compliers and noncompliers using baseline covariates. Next, we extend current methods by deriving standard errors for these estimates. We demonstrate these methods using an IV known as tendency to operate (TTO) from health services research. 
\end{abstract}


\vspace{1cm}

\noindent\textit{Abbreviations}: IV, instrumental variable; TTO, Tendency to Operate;  LATE, Local Average Treatment Effect 

\clearpage

\doublespacing

\section{Introduction} 
\label{sec:intro}

The instrumental variable (IV) method has become a key study design in health services research and epidemiology. An IV is a device that encourages exposure to a treatment of interest. IVs occur in both randomized experiments with noncompliance and observational studies. IV methods are viewed as valuable, since valid causal inferences can be drawn even in the presence of unobserved confounding \cite{angrist1996identification,baiocchi2014instrumental}. The IV method, however, requires strong assumptions. To improve the IV study design, an expanding literature in epidemiology has focused on developing guidelines, diagnostics, and falsification tests for evaluating these assumptions.\cite{baiocchi2012near,swanson2013commentary,Yang:2013,glymour2012credible,kang2013causal,pizer2016falsification,keeleexrest2018,bransonkeele2020}. Here, we further contribute to this literature by developing new methods for profiling compliers.

In the IV study design, $D$ denotes treatment exposure and $Y$ is the outcome of interest. The instrument, $Z$, denotes that some units were provided encouragement for treatment exposure. Using the IV method, investigators use $Z$ to instrument $D$, which allows them to identify the effect of $D$ on $Y$ even if unobserved variables confound the relationship between $D$ and $Y$. Many IV studies invoke a framework where study subjects are classified into four different strata: compliers, always-takers, never-takers, and defiers.\cite{angrist1996identification} Compliers are patients who are exposed to the treatment $(D=1)$ if and only if they were encouraged by the instrument $(Z=1)$.  Always-takers are exposed to $D$ regardless of their value of $Z$, and never-takers are not exposed to $D$ regardless of their value of $Z$. Defiers are only exposed to $D$ if not encouraged, and not exposed if encouraged. The presence of defiers is ruled out by invoking an assumption known as monotonicity.\cite{angrist1996identification} Under monotonicity, the IV estimate---often referred to as the local average treatment effect (LATE)---is identified and describes the causal effect of $D$ on $Y$ for compliers only.

When an IV analysis is focused on the complier subpopulation, it invariably raises the question of whether the LATE differs from the average treatment effect (ATE)---the effect of $D$ on $Y$ for the entire study population. This question becomes more acute when the proportion of compliers is relatively small. Several authors and guidelines on IV analyses recommend complier profiling \cite{angrist2008mostly,swanson2013commentary,baiocchi2014instrumental}. Profiling refers to calculating descriptive statistics for the complier population. The complier profile can then be compared to the overall patient population profile, as well as the always-taker and never-taker subpopulations. If profiling reveals that the complier subpopulation is different with respect to observable covariates that are likely to be predictive of treatment effect magnitude, this suggests that the LATE may not be close to the ATE. Moreover, comparing never-takers to compliers allows the investigator to reason about the sub-population that cannot be induced to take the treatment. One example, where profiling compliers is relevant, are studies that leverage surgeon tendency to operate (TTO) as an instrument for estimating the LATE of operative management (i.e. the treatment) relative to non-surgical care. In this context, compliers are patients for whom surgeons with a high TTO choose operative management and surgeons with a low TTO choose non-surgical care, while never-takers are patients for whom doctors never choose surgery. Below we analyze such data and find that compliers are substantially more likely to suffer from sepsis at admission than never-takers. This difference in a key covariate that has been shown to moderate the effect of surgery outcomes should caution investigators to assume that the LATE generalizes to the ATE. We further discuss this application below. 

This paper starts by reviewing a simple method for profiling subpopulations in IV studies developed by Marbach and Hangartner \cite{marbach2020profiling}. This method of profiling provides descriptive statistics for not only compliers but also always-takers and never-takers. In this paper, we extend this profiling method in several key ways. We show that their estimator is consistent and normally distributed and we derive analytic standard errors. We also outline the assumptions necessary for applying the method to contexts where the IV is not randomly assigned but ignorable conditional on covariates. Our profiling method allows for simple visual comparisons of the characteristics of compliers, non-compliers and the overall study population and incorporates measures of statistical uncertainty. In Section~\ref{sec:notation}, we outline notation and necessary assumptions for estimating the LATE and for IV profiling. In Section~\ref{s:profilest}, we review identification conditions and estimation methods for profiling compliers. Next, we derive standard errors and use simulation to demonstrate coverage of confidence intervals. In Section~\ref{s:application}, we return to our IV application that uses surgeon TTO as an instrument for emergency surgery and apply our method. In Section~\ref{s:conclusion}, we conclude and point researchers to the software that implements these methods.

\section{Notation and IV Assumptions} 
\label{sec:notation}

First, we define notation. Let $D$ be the binary treatment exposure, $Y$ be the outcome, $Z$ be a binary instrument, and $\bm{X} = (X_1,\ldots,X_K)$ be a set of baseline covariates measured before $Z$ and $D$ are assigned. The covariates in $\bm{X}$ describe the units and  will be used to profile the compliers in the population (e.g., age, gender, pre-existing co-morbidities). To simplify notation, below, we will will focus on a single profiling covariate, $X$. 

For $z = 0,1$ and $d = 0,1$, let $Y{(z,d)}$ be the potential outcome of $Y$ when $Z=z$ and $D=d$, and let $D{(z)}$ be the potential outcome of $D$ when $Z = z$. The observed values $(Y, D, Z)$ are related to the potential outcomes through the following equations:
\[
D = D{(Z)} = Z D{(1)}  + (1 - Z) D{(0)},
\]
\[
Y  = Y{(Z,D)} =  Y{\left(Z, D{(Z)}\right)} = Z Y{\left(1,D{(1)} \right)} + (1 - Z) Y{\left(0,D{(0)} \right)}.
\]

This notation implicitly assumes the stable unit treatment value assumption (SUTVA) \cite{rubin1980randomization}. SUTVA implies that 1) the levels of $D$ (1 and 0) adequately represent all versions of the treatment, and 2) each subject's outcome is not affected by other subjects' exposures.  The first component of SUTVA is often referred to as the consistency assumption in epidemiology.\cite{Schwartz:2012} In addition to SUTVA, a standard IV analysis requires three additional core assumptions. The subset of assumptions that are necessary to identify the LATE, but are not needed for profiling, are labelled with ``$*$''.
\begin{description}
\item[(A1) Relevance:] $\mathbb{E}[D|Z=1] \ne \mathbb{E}[D|Z=0]$.
\end{description}
A1 implies that the instrument has a nonzero effect on treatment exposure. For both, estimation of treatment effects and profiling, the instrument must have a discernible effect on patients' treatment exposure.
\begin{description}
\item[(A2*) Exclusion restriction:] $Y(z,d) = Y(z',d)$ for any $z$,
  $z'$, and $d$. 
\end{description}   

A2* implies that the potential outcomes only depend on the instrument $Z$ through its effect on the exposure $D$. This assumption allows us to rewrite the potential outcomes as $Y(0,d) = Y(1,d) = Y(d)$. Note, however, that this assumption is not needed for  profiling, since profiling does not leverage information about the outcome $Y$.
\begin{description} 
\item[(A3*) Random or effective (as-if) random assignment I:] $Y(z,d), D(z) \independent Z $ for all $z = 0,1$ and $d = 0,1$. 
\end{description} 
This implies that there are no unmeasured confounders for $Y(z,d)$ and $Z$, and $D(z)$ and $Z$, respectively. Since profiling does not require any assumptions about the relationship between the instrument and the outcome, we will refine A3* below. A3* will hold by design in randomized trials with noncompliance. If $Z$ is not randomly assigned, additional covariate adjustment might be necessary, which we discuss below. In both cases, A3* can be probed using falsification tests.\cite{bransonkeele2020} 

In addition to assumptions A1, A2* and A3*, many IV analyses whose goal is to estimate the LATE invoke the following monotonicity assumption:

\begin{description}
\item[(A4) Monotonicity:] $D(1) \geq D(0) $. 
\end{description}
The monotonicity assumption implies that there are no defiers, i.e. units who would always do the opposite of what the instrument encourages them to do. If defiers are not present, investigators can interpret the effect of $D$ on $Y$ for the (local) average treatment effect among compliers. Notably, A4 holds if control subjects are unable to access the treatment in the absence of encouragement. That is, A4 holds by design if $\mathbb{P}(D = 0 | Z=0) = 1$.

In the IV design we consider, study units can be classified into four strata based on combinations of treatment exposure $D$ and instrument $Z$. Table~\ref{tab:strata} contains the study population stratified by $D$ and $Z$. These strata are: always-takers $(D(0)=D(1)=1)$, never-takers $(D(0)=D(1)=0)$, compliers $(D(0)<D(1))$ and defiers $(D(0)>D(1))$.  Under A4, we rule out the presence of defiers, which allows us to directly identify the compliance status of the noncomplying patients located on the off-diagonal. We call the never-takers (always-takers) on the off-diagonal ``observable'' never-takers (always-takers) and the never-takers (always-takers) on the main diagonal ``non-observable'' never-takers (always-takers).

\begin{table}[ht!]
\begin{center}
\caption{Stratification of the study population by treatment ($D$) and instrument ($Z$) values.\label{tab:strata}}
\begin{tabular}{llcc}
\toprule
&	& \multicolumn{2}{c}{$Z$} \\ 
\cmidrule(lr){3-4}
 & & 0 & 1 \\ 
\midrule 
\multirow{2}{*}{$D$} & 0 & Compliers \& Never-takers & Never-takers \& \sout{Defiers}   \\ 
	& 1 & Always-takers  \& \sout{Defiers}   & Compliers \& Always-takers \\ 
\bottomrule
\end{tabular}
\captionsetup{justification=raggedright,
singlelinecheck=false
}
\end{center}
\end{table}


Researchers interested in identifying the LATE typically invoke assumptions A1, A2*, A3* and A4.\cite{angrist1996identification,baiocchi2014instrumental}. In contrast, profiling does not require the exclusion restriction but instead that the instrument is independent (as-if randomly assigned) from the profiling covariate $X$. As such, profiling requires us to change A3* as follows:

\begin{description}
\item[(A5) Random or effective (as-if) random assignment II:] $X, D(z) \independent Z$ for all $z = 0,1$. 
\end{description}
This implies that there are no unmeasured confounders for $D(z)$ and $Z$, and $X$ and $Z$, respectively.  In contrast to assumption A3*, we do not need to make assumptions about the relationship of $Z$ and $Y$. In that sense, A5 is weaker than A3*. However, profiling requires independence between the profiling covariate $Z$ and $X$. In that sense, A5 is stronger than A3*. In sum, for profiling alone, researchers only need to assume A1, A4, and A5.

In research designs where the instrument is not randomly assigned, researchers might need to adjust for (additional) covariates to make assumption A5 hold. Let $\bm{W} = (W_1,\ldots,W_L)$ be a set of covariates measured before $Z$ and $D$ are assigned. $\bm{W}$ contain those confounders (e.g., age, gender, race/ethnicity) that are necessary to adjust for such that $D(z), X \independent Z | \bm{W} $ holds for all $z = 0,1$. Note that the same covariate used for profiling can also be used for adjustment, i.e. $ X \cap \bm{W}$ may be nonempty.  The important implication of A5 is that complier profiling should be done conditional on covariates in IV designs where the IV is not randomly assigned. To conduct complier profiling conditional on covariates, researchers can use, for example, weighting-based semiparametric instrumental variable estimation.\cite{abadie2003larf}. An alternative approach is to use matching to adjust for IV-outcome confounders; see Baiocchi et al.\cite{baiocchi2012near} and Keele et al.\cite{Keele:2016c} for examples. With matched data, complier profiling can be done using estimators unconditional of $\bm{W}$. We apply the matching approach to the observational study application presented below. Therefore, and to simplify notation, we discuss identification and estimation of the complier covariate mean without explicitly conditioning on $\bm{W}$ next.

\section{Profiling Complier Strata in IV Studies} 
\label{s:profilest}

In the next section, we review methods for profiling compliers. Building on Abadie's $\kappa$-weights used for semiparametric instrumental variable estimation \cite{abadie2003larf}, the concept of complier profiling was first outlined by Angrist and Pischke \cite{angrist2008mostly}. Baiocchi et al. extended complier profiling and introduced the technique to health service research and epidemiology \cite{baiocchi2014instrumental}.  Marbach and Hangartner extended profiling to the always-taker and never-taker subpopulations, developed graphical displays for the results, and provided software \cite{marbach2020profiling}. In section~\ref{var}, we present new results for variance estimation.

\subsection{Complier Profile Estimation}
\label{ident} 
  
Under assumptions A1 and A5, observable and non-observable always-takers have the same distribution for $X$ such that $\mathbb{E}[ X | D(0)=D(1)=1, Z ] = \mathbb{E}[ X | D(0)=D(1)=1]$. The same equality holds for never-takers: $\mathbb{E}[ X | D(0)=D(1)=0, Z ] = \mathbb{E}[ X |  D(0)=D(1)=0]$. Under assumptions A1, A4, and A5, the profiling covariate mean for always-takers and never-takers, respectively, is
\begin{equation} \label{eq1}
\mathbb{E}[ X |  D(0)=D(1)=1 ] = \mathbb{E}[ X | D=1, Z=0]
\end{equation}
and
\begin{equation} \label{eq2}
\mathbb{E}[ X | D(0)=D(1)=0 ] = \mathbb{E}[ X | D=0, Z=1].
\end{equation}
This implies that investigators can readily estimate descriptive statistics for these two strata in Table~\ref{tab:strata} substituting population moments with sample analogs.

In contrast, compliers, who are located on the main-diagonal of Table~\ref{tab:strata} cannot be identified at the individual level, since compliers assigned to $Z=0$ have the same treatment and instrument values as never-takers assigned to $Z=0$, and compliers assigned to $Z=1$ are mixed with the always-takers assigned to $Z=1$. However, by subtracting $\mathbb{E}[ X | D(0)=D(1)=1 ] $ and $\mathbb{E}[ X | D(0)=D(1)=0 ]$, each weighted by the size of the corresponding strata, from the population mean $\mathbb{E}[ X ]$, investigators can identify the covariate mean for compliers: $\mathbb{E}[ X | D(0)<D(1) ]$. Under assumptions A1, A4 and A5 and substituting the results from equations~(\ref{eq1}) and (\ref{eq2}), the covariate mean of compliers is:

\begin{equation} \label{eq3}
\begin{split}
\mathbb{E}[ X | D(0)<D(1) ]= \big (\mathbb{E}[ X ]  - & \mathbb{E}[ X |D=0, Z=1] \mathbb{P}[D=0| Z=1 ]  \\
						- & \mathbb{E}[ X | D=1, Z=0] \mathbb{P}[D=1 | Z=0 ]\big ) \\ 
						& (\mathbb{E}[D|Z=1]-\mathbb{E}[D | Z=0])^{-1}. 
\end{split}
\end{equation}
In order to use equation \eqref{eq3} for estimation, we plug in  $\mathbb{E}[D|Z=1]=\mathbb{P}[D=1|Z=1]=1-\mathbb{P}[D=0|Z=1]$ and then expand all conditionals to obtain
\begin{equation} \label{eqLeo}
\begin{split} 
\mathbb{E}[ X | D(0)<D(1) ]= \big ( \mathbb{E}[ X ]  - & \frac{\mathbb{E}[ X \mathbbm{1}_{\{D=0\}} \mathbbm{1}_{\{Z=1\}}]}{\mathbb{P}[Z=1]} - \frac{\mathbb{E}[ X \mathbbm{1}_{\{D=1\}} \mathbbm{1}_{\{Z=0\}}]}{1-\mathbb{P}[Z=1 ]}\big ) \\
& \left(1 - \frac{\mathbb{P}[D=0,Z=1]}{\mathbb{P}[Z=1]}-\frac{\mathbb{P}[D=1,Z=0]}{1-\mathbb{P}[Z=1]}\right)^{-1}. 
\end{split}
\end{equation}

Based on $N$ i.i.d. realizations from $(Z,X,D)$, each of the six terms in equation \eqref{eqLeo} can be estimated as a simple sample mean. The estimators for the three expectations terms are\\
$\hat{\mu}=\frac{1}{N}\sum_{i=1}^N X_i$, $\hat{\mu}_{vnt}=\frac{1}{N}\sum_{i=1}^N Z_i (1-D_i) X_i$, $\hat{\mu}_{vat}=\frac{1}{N}\sum_{i=1}^N (1-Z_i) D_i X_i$, \\
\noindent and the ones for the three event probabilities are \\ 
 $\hat{\pi}_{vnt}=\frac{1}{N}\sum_{i=1}^N Z_i (1-D_i)$, $\hat{\pi}_{vat}=\frac{1}{N}\sum_{i=1}^N (1-Z_i) D_i$, $\hat{\pi}_{z}=\frac{1}{N}\sum_{i=1}^N Z_i$. The resulting plug-in estimator for the complier covariate mean is
\begin{align} \label{eq4}
\hat{\mu}_{co}
&= \left(\hat{\mu} - \frac{\hat{\mu}_{vnt}}{\hat{\pi}_z}- \frac{\hat{\mu}_{vat}}{1-\hat{\pi}_{z}}\right) /\left(1 - \frac{\hat{\pi}_{vnt}}{\hat{\pi}_z} - \frac{\hat{\pi}_{vat}}{1-\hat{\pi}_z}\right) \,\,\, .
\end{align}

Analysts can use this estimator to estimate the covariate means of the relevant complier strata. Comparing the covariate means of compliers and non-compliers can reveal the extent that these group compositions are heterogeneous. However, to draw meaningful comparisons, analysts also require measures of statistical uncertainty. We take up this issue next. 

\subsection{Variance Estimation}
\label{var} 

Marbach and Hangartner suggest the bootstrap for obtaining measures of statistical uncertainty for the complier covariate mean in equation~(\ref{eq4}). Here, we provide two contribution in terms of variance estimation. First, we show that the estimator for the complier mean is consistent and asymptotically normal using the Delta method. Second, as a corollary, we derive an estimator for the complier mean standard error. One advantage of analytical standard errors is that they are computationally less intensive than the bootstrap. 
 
While the full proof for consistency is provided in the Appendix, here we give a brief sketch. We group the estimators for covariate means and event probabilities defined in Section~\ref{ident} as elements in a vector: $\hat{\boldsymbol{\beta}}= (\hat{\mu},\hat{\mu}_{vnt}, \hat{\mu}_{vat}, \hat{\pi}_{vnt},\hat{\pi}_{vat},\hat{\pi}_{z})$. We apply the multivariate Central Limit Theorem to $\hat{\boldsymbol{\beta}}$ and then use this to apply the Delta Method to equation \eqref{eq4}.  As a result we obtain 

\begin{equation}
\sqrt{N}(f(\boldsymbol{\hat{\beta}}) - \mu_{co}) \xrightarrow{d} \mathcal{N}(0,\nabla f(\boldsymbol{\beta})^T \Sigma \nabla f(\boldsymbol{\beta})) \,\,\, ,
\label{a.dis.equ}
\end{equation}

where $f$ is the formula from equation \eqref{eq4}, $\nabla f(\boldsymbol{\beta})$ is a gradient vector and $\Sigma$ the covariance matrix from the multivariate Central Limit Theorem for $\hat{\boldsymbol{\beta}}$ (see the proof for Theorem 1 in the Appendix for details). Let $\hat{\Sigma}$ denote the sample covariance matrix corresponding to $\Sigma$. Under the mild additional assumption that $X$ have finite fourth moment, we find that the plug-in estimator $\nabla f(\boldsymbol{\hat{\beta}})^T \hat{\Sigma} \nabla f(\boldsymbol{\hat{\beta}})$ for the asymptotic variance $\nabla f(\boldsymbol{\beta})^T \Sigma \nabla f(\boldsymbol{\beta})$ from equation~(\ref{a.dis.equ}) is consistent. We can therefore compute asymptotically valid standard errors for our covariate mean estimator $\hat{\mu}_{co}$ as

\begin{equation}
SE(\hat{\mu}_{co})=\sqrt{ \frac{1}{N} \nabla f(\boldsymbol{\hat{\beta}})^T \hat{\Sigma} \nabla f(\boldsymbol{\hat{\beta}})}.
\label{coro}
\end{equation}

With the analytical variance estimator provided by this corollary, investigators need not use the bootstrap to obtain measures of statistical uncertainty for strata means.

\subsection{Simulation Study} 
\label{ss:profilese}
         
Next, we conduct a series of Monte Carlo simulations to verify the consistency of the plug-in variance estimator. We also examine whether the corresponding 95\% confidence intervals achieve the nominal coverage rate. For each simulation, we generate a dataset with $N$ units that are randomly assigned to one of the three strata: complier, always-taker or never-taker. For each unit, we simulate a realized value of the instrument, the corresponding observed treatment status, and a value for the continuous covariate. In the first variant of the simulations, we fix all parameters of the data generating process. The strata are of equal size and three quarters of the sample receive the encouragement. The covariate is drawn from three normal distributions with strata-specific parameters $\mu_{co}=2, \sigma_{co}=0.5$, $\mu_{nt}=1, \sigma_{nt}=1$, and $\mu_{at}=0.5, \sigma_{at}=2$, respectively. 

In the second variant of the simulations, we vary the probability that a unit belongs to a strata by drawing from a uniform Dirichlet density under the constraint that no strata share is lower than 0.1. The binary instrument is drawn from a Bernoulli distribution with a success probability that is drawn from a uniform density in a range from 0.1 to 0.9. To simulate the covariate distribution, we again draw from three strata-specific normal densities. The mean parameters for these normal densities are drawn from a uniform distribution with support -2 to 2, and the standard deviations from a uniform density with support 0.25 to 2. This setup is identical to the Monte Carlo experiments in Marbach and Hangartner\cite{marbach2020profiling}. In both simulation variants, we used 13 different sample sizes ranging from $N$=500 to $N$=24,000, and there were 1,000 simulations for each sample size. 

Table \ref{tab:mc} provides the results from the first set of simulations. Column 1 reports the sample size, columns 2-3 the root mean squared error (RMSE) for the bootstrap and the plug-in estimate, and columns 4-5 show the corresponding 95\% coverage. As expected, both variance estimates converge with increasing sample size and the coverage rate is close to the nominal 95\% rate. Across all sample sizes, the RMSE for the plug-in estimate tends is smaller than the RMSE for the bootstrap estimate, suggesting a higher efficiency for the plug-in variance estimator as compared to the bootstrap estimator. 

\begin{table}[ht]
\centering
\caption{Monte Carlo simulations results across different sample sizes with the root mean squared error scaled by sample size for bootstrap and plug-in estimate (columns 2-3), and the coverage rate of the 95\% confidence interval based on the bootstrap and plug-in estimate (columns 4-5).\label{tab:mc}}
\begin{tabular}{crrrrr}
\toprule
	&  \multicolumn{2}{c}{Variance} & \multicolumn{2}{c}{Coverage rate} \\
\cmidrule(lr){2-3}
\cmidrule(lr){4-5}
N & Bootstrap & Plug-in & Bootstrap & Plug-in \\ 
\midrule
500 & 100.7089 & 41.3911 & 0.951 & 0.947 \\ 
750 & 134.7876 & 35.0345 & 0.953 & 0.942 \\ 
1000 & 26.5930 & 22.0742 & 0.970 & 0.966 \\ 
1500 & 20.4587 & 18.3057 & 0.948 & 0.946 \\ 
2000 & 16.1971 & 14.8250 & 0.952 & 0.945 \\ 
3000 & 13.3639 & 12.5435 & 0.944 & 0.945 \\ 
4000 & 10.9481 & 10.2017 & 0.961 & 0.957 \\ 
6000 & 8.6696 & 7.9722 & 0.954 & 0.953 \\ 
8000 & 7.7589 & 7.1140 & 0.950 & 0.949 \\ 
12000 & 6.4108 & 5.7230 & 0.951 & 0.945 \\ 
16000 & 5.5209 & 4.9234 & 0.940 & 0.940 \\ 
20000 & 5.0438 & 4.2608 & 0.955 & 0.953 \\ 
24000 & 4.7506 & 4.0273 & 0.958 & 0.954 \\ 
 \bottomrule
\end{tabular}
\end{table}

Figure~\ref{fig:mc} reports the results from the second set of simulations focusing on the plug-in variance estimate. Exploring a larger segment of the parameter space, we find that the 95\% confidence intervals based on the plug-in  estimate display estimated coverage rates between 94\% and 96\% across all sample sizes.

\begin{figure}[H]
\centering
\includegraphics[width=0.6\textwidth]{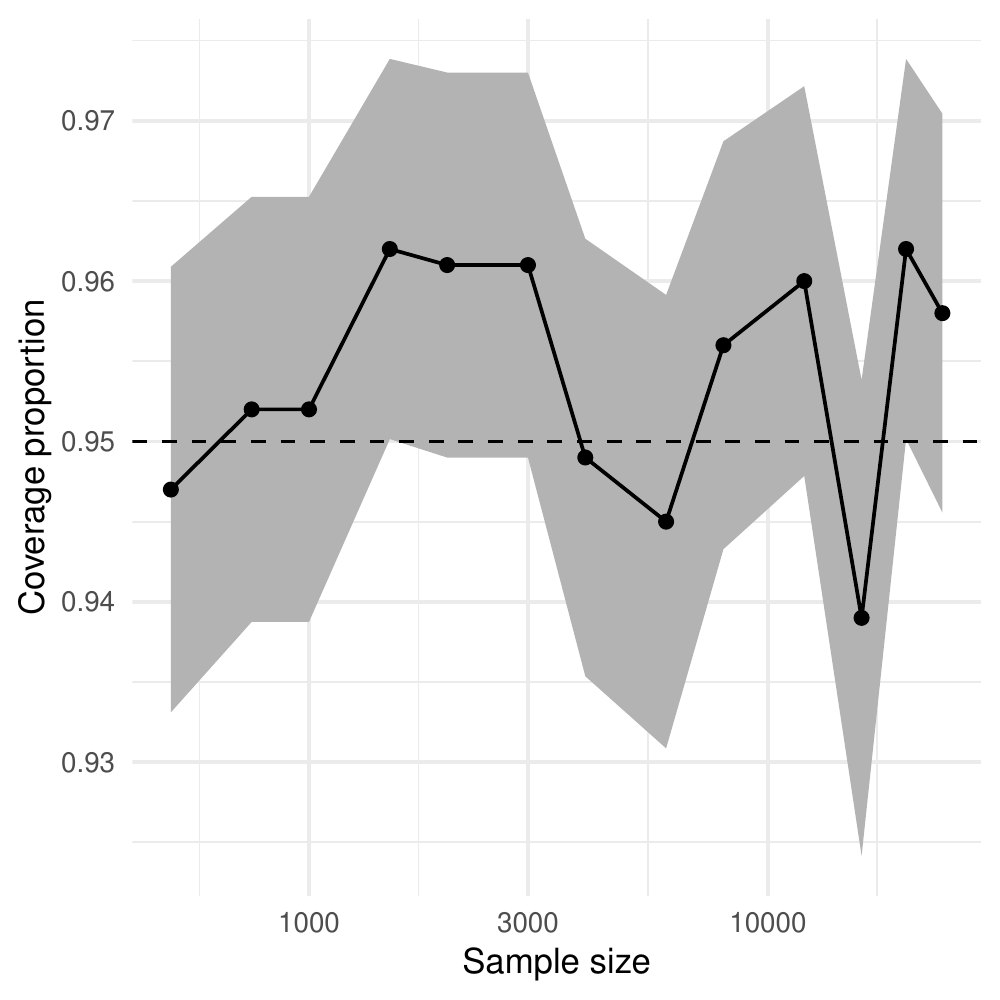}
\captionsetup{justification=raggedright,
singlelinecheck=false
}
\caption{The black lines shows the coverage rate for the 95\% confidence interval of the complier mean across 13 Monte Carlo experiments with varying sample sizes based on the plug-in variance estimator. The gray ribbon visualizes the 95\% confidence intervals of the coverage rate based on the Monte Carlo simulation error. \label{fig:mc}} 
\end{figure}

The Monte Carlo simulation confidence intervals (in gray) for the coverage proportion overlap with the nominal coverage rate. This suggests that the plug-in variance estimator has accurate coverage for all sample sizes used in the simulations. 

\section{Application: Emergency general surgery and physician preference for operative management} 
\label{s:application}

In this section, we illustrate the the use of complier profiling using a recent study by Keele et al.\ \cite{keeleegsiv2018}. In that study, the authors leverage surgeons tendency to operate (TTO) as an IV for whether patients received emergency general surgery (EGS) for a set of pre-defined emergency medical conditions. They measured TTO as the percentage of similar emergency cases where the surgeon chose to operate in the past. TTO is one example of a physician's treatment preferences, which is an often used instrument in clinical studies \cite{brookhart2007preference,keeleegsiv2018}. 

The data is based on all-payer hospital discharge data from New York, Florida and Pennsylvania in 2012-2013, linked to the American Medical Association (AMA) Physician Masterfile to join data on patient claims with surgeons. The data include patient sociodemographic and clinical characteristics including a measure of patient frailty, an indicator for severe sepsis, and 31 indicators for comorbidities based on Elixhauser indices \cite{elixhauser1998comorbidity}. The study population was restricted to patients admitted for inpatient emergency care, urgently, or through the emergency department with a diagnosis of an acute colitis and diverticulitis. Surgeons were excluded if they could not be identified in the AMA Masterfile, did not meet criteria for general surgery training, did not attend an allopathic program, or did not train within the United States. To illustrate our proposed method for profiling, we replicate the results used in Fogarty et al.\cite{fogarty2020ivsens} that focused on patients with either colitis and diverticulitis. In their analysis, they adjusted for possible IV-outcome confounders using a variant of matching that combines near-far matching with refined covariate balance \cite{Pimentel:2015a,keeleicumatch2019}. More specifically, patients were exactly matched on hospital such that across-hospital differences could not bias the analysis. They were also exactly matched on an indicator for sepsis. For the remaining covariates, they minimized the total of the within-pair distances on covariates as measured by the Mahalanobis distance and applied a caliper to the propensity score. After matching, there are 3,034 pairs of patients. In each pair, patients are matched to have care from a surgeon with a low and high TTO, where TTO is divided at the sample median. Conditional on these covariate adjustments, we expect that the instrument is as-if randomly assigned with respect to the treatment, outcome, and profiling variables. Primary outcomes in the study were the presence of an in-hospital complication and length of stay in days. The estimated LATE for complications is 0.17 (95\% CI: 0.14, 0.21) and for length of stay is 5.2 (95\% CI: 4.1, 6.2). Next, we use complier profiling to gauge to what extent these results generalize to the larger study population.

First, we review the definitions of the key strata in this study. Here, compliers are patients for whom surgeons with a high TTO choose operative management and surgeons with a low TTO choose non-surgical care, and never-takers are patients for whom surgeons never choose operative management, and always-takers are patients for whom surgeons always use operative management. Based on baseline covariates, we ask whether compliers appear to be quite different from always-takers and never-takers. Applying the estimator outlined in Sections \ref{ident} and \ref{var} , we can profile the patients in the study sample in terms of their covariates. In particular, we focus on patients' socio-demographic and clinical characteristics. For analysis, we use the dedicated \sf{R} package ``ivdesc''. Figure~\ref{fig:fig1} visualizes the means and associated confidence intervals for compliers, always-takers, never-takers, and the entire sample for a subset of available covariates. Table 1 in the appendix contains the underlying numerical estimates.

\begin{figure}[htbp]
  \centering
  \includegraphics[scale=.99]{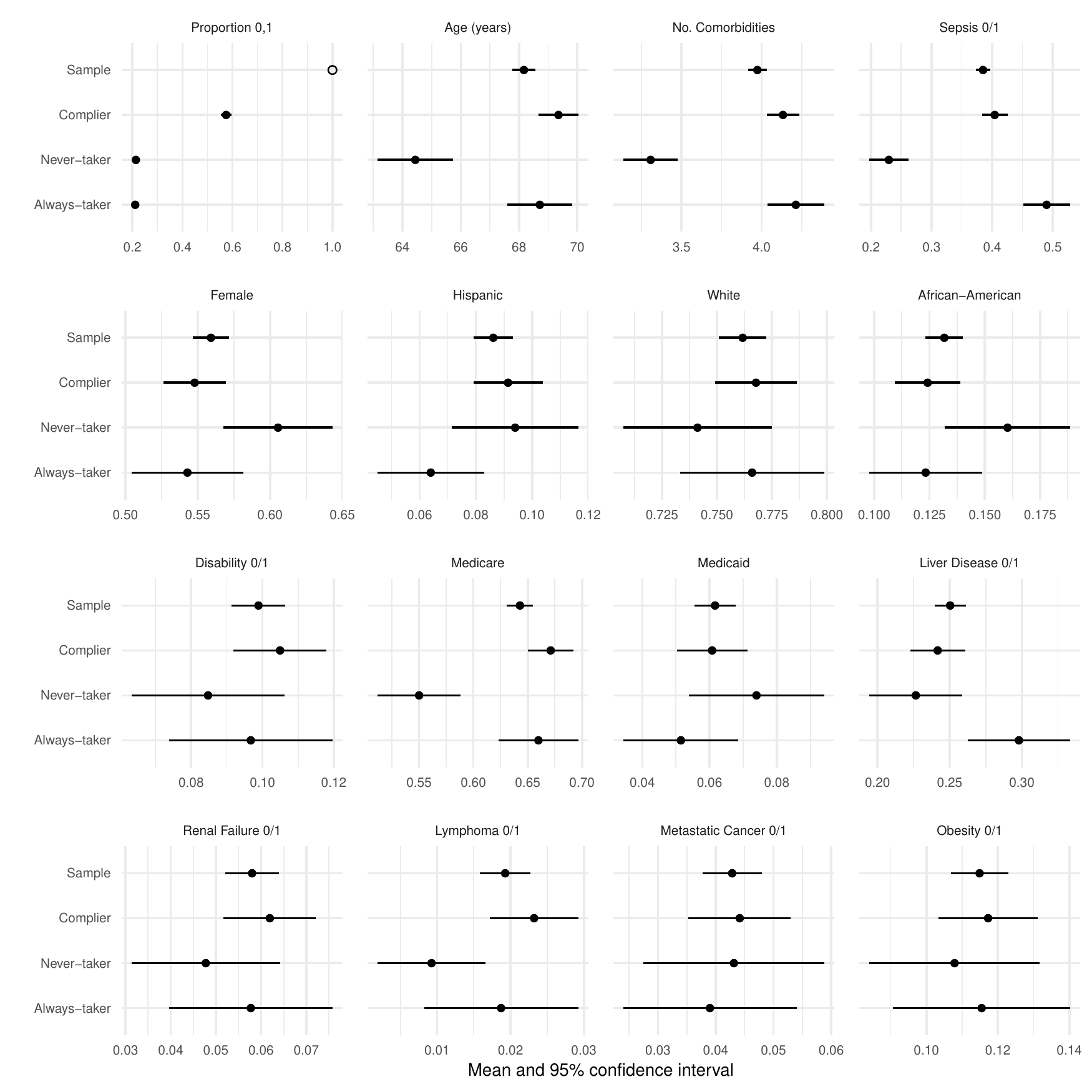}
  \caption{Instrumental Variable Subpopulation Descriptive Statistics for Select Covariates In The EGS Application.}
\label{fig:fig1}
\end{figure}

The first panel in Figure~\ref{fig:fig1} displays the sample proportions of never-takers (21\%), always-takers (21\%), and compliers (57\%). The remaining panels show the covariate means and corresponding 95\% confidence intervals based on the analytical variance estimators for the three strata as well as the overall sample mean. Several relevant differences emerge. Compared to never-takers, we find that compliers are older (69.35 versus 64.44 years), more likely to have Medicare insurance (67\% versus 55\%), more likely to suffer from sepsis (40\% versus 23\%) and have a greater number of comorbidities (4.13 versus 3.31). Compared to always-takers, compliers are less likely to suffer from sepsis (40\% versus 49\%) and liver disease (24\% versus 30\%), and more likely to be hispanic (9\% versus 6\%). We find only minor differences across the three strata on covariates such as disability, renal failure, metastatic cancer, and obesity. Fogarty et al.\cite{fogarty2020ivsens} studied the extent to which the effect of surgery was heterogenous. They found that sepsis was a key effect modifier, but the rest of the covariates were not significant effect modifiers. As such, differences on this covariate imply that investigators should be cautious when attempting to extrapolate the LATE estimated for compliers to the entire study sample.

\section{Conclusion} 
\label{s:conclusion}

Our study shows how to profile compliers and non-compliers in the context of an IV application that uses surgeon preferences for operative care to instrument for emergency surgery. We extend the profiling method introduced by Marbach and Hangartner by showing that their estimator is consistent and normally distributed, deriving analytic standard errors, and discussing the assumptions necessary for applying the method to contexts where the IV is not randomly assigned but ignorable conditional on covariates. After using near-far matching to adjust for observed confounders, we find relevant differences in sociodemographic and clinical characteristics between compliers, always-takers, and never-takers. Since at least one of these characteristics, sepsis, has been shown to influence the size of the effect of surgery on outcomes, our profiling results should caution investigators to generalize the surgery effect estimated for compliers to the rest of the study sample. In line with recent guidelines that recommend complier profiling to be part of any IV analysis\cite{angrist2008mostly,swanson2013commentary,baiocchi2014instrumental}, we hope the methods presented here, as well as the accompanying software packages ``ivdesc'' for \sf{R} and \sf{STATA} that implements them, to be useful tools for applied researchers.

%

\clearpage
\singlespacing

\bibliographystyle{ama}
\bibliography{ref}

\newpage

\begin{appendix}

\begin{center}
{\LARGE Appendix\\[2ex]
\Large ``Profiling Compliers in Instrumental Variables Designs''
}\\
\normalsize 
\vspace{0.5cm}
\end{center}
\vspace{2cm}

\section{Proofs} 
\label{delta}

We first derive the asymptotic distribution of the complier mean estimator with the delta method. As a corollary, we obtain a consistent estimator for the estimator's asymptotic variance. 



\begin{theorem} \label{a.dis}
Consider a random vector $(Z,S,X(S),D(Z,S))$, with random instrument\\ $Z \sim \mathrm{Bernoulli}(\pi_z), \pi_z \in (0,1)$, strata $S \in \{nt,at,co\}$ from a categorical distribution with event probabilities $\pi_{nt},\pi_{at}, \pi_{co}=1-\pi_{nt}-\pi_{at}\neq 0$, covariate $X$ from a mixed distribution with three strata dependent on $S$ with means $\mu_{at},\mu_{nt}$ and $\mu_{co}$ as well as finite variances $\sigma_{nt},\sigma_{at}$ and $\sigma_{co}$, respectively, and treatment 
\begin{equation*}
	D(Z,S) = 
	\begin{cases}
		0 &\text{if $S=nt$}\\
		1 &\text{if $S=at$}\\
		Z &\text{if $S=co$}.
	\end{cases}
\end{equation*}
Suppose that $Z$ is independent of $(X,S)$ and
%
let $\boldsymbol{\hat{\beta}}= (\hat{\mu},\hat{\mu}_{vnt}, \hat{\mu}_{vat}, \hat{\pi}_{vnt},\hat{\pi}_{vat},\hat{\pi}_{z})$, with 

\begin{enumerate}
	\item[] i) $\hat{\mu}=\frac{1}{N}\sum_{i=1}^N X_i$, ii) $\hat{\mu}_{vnt}=\frac{1}{N}\sum_{i=1}^N Z_i (1-D_i) X_i$, iii) $\hat{\mu}_{vat}=\frac{1}{N}\sum_{i=1}^N (1-Z_i) D_i X_i$,
	\item[] iv) $\hat{\pi}_{vnt}=\frac{1}{N}\sum_{i=1}^N Z_i (1-D_i)$,
	v) $\hat{\pi}_{vat}=\frac{1}{N}\sum_{i=1}^N (1-Z_i) D_i$ and
	vi) $\hat{\pi}_{z}=\frac{1}{N}\sum_{i=1}^N Z_i$
\end{enumerate}
based on N i.i.d. realizations from $(Z,X,D)$. Consider the estimator

\begin{align*}
\hat{\mu}_{co}
&=f(\boldsymbol{\hat{\beta}})=f(\hat{\mu},\hat{\mu}_{vnt}, \hat{\mu}_{vat}, \hat{\pi}_{vnt},\hat{\pi}_{vat},\hat{\pi}_{z}) \\
&= \left(\hat{\mu} - \frac{\hat{\mu}_{vnt}}{\hat{\pi}_z}- \frac{\hat{\mu}_{vat}}{1-\hat{\pi}_{z}}\right) /\left(1 - \frac{\hat{\pi}_{vnt}}{\hat{\pi}_z} - \frac{\hat{\pi}_{vat}}{1-\hat{\pi}_z}\right).
\end{align*}

Then 

\begin{equation*}
\sqrt{N}(f(\boldsymbol{\hat{\beta}}) - \mu_{co}) \xrightarrow{d} \mathcal{N}(0,\nabla f(\boldsymbol{\beta})^T \Sigma \nabla f(\boldsymbol{\beta})),
\end{equation*}

where $\boldsymbol{\beta}=\mathbb{E}[\mathbf{C}]$ and $\Sigma=\mathrm{Cov}(\mathbf{C},\mathbf{C})$, with 

\begin{equation*}
\mathbf{C} =
\begin{pmatrix}
X\\
Z(1-D)X\\
(1-Z)DX\\
Z(1-D)\\
(1-Z)D\\
Z\\
\end{pmatrix}.
\end{equation*}

\end{theorem}

\begin{proof}
	
		We first show that $\sqrt{N}(\boldsymbol{\hat{\beta}} - \boldsymbol{\beta}) \xrightarrow{d} \mathcal{N}(0,\Sigma)$.
	The estimator $\boldsymbol{\hat{\beta}}$ can be rewritten as $\boldsymbol{\hat{\beta}}=\frac{1}{N} \sum_{i=1}^N \mathbf{C}_i$, with 
	
	\begin{equation*}
	\mathbf{C}_i =
	\begin{pmatrix}
	X_i\\
	Z_i(1-D_i)X_i\\
	(1-Z_i)D_iX_i\\
	Z_i(1-D_i)\\
	(1-Z_i)D_i\\
	Z_i\\
	\end{pmatrix}.
	\end{equation*}
	Each individual strata $X(S),S\in\{nt,at,co\}$ has a finite mean and variance and as a result so does $X$. Hence, the $\mathbf{C}_i$ are i.i.d. with finite mean and variance. We can thus invoke the central limit theorem and conclude that $\sqrt{N}(\boldsymbol{\hat{\beta}} - \boldsymbol{\beta}) \xrightarrow{d} \mathcal{N}(0,\Sigma)$, where $\Sigma=\mathrm{Cov}(\mathbf{C},\mathbf{C})$. 
	
	By construction of $D$ and the assumption of random treatment assignment			
	\begin{align*}
	\pi_{vnt}=\mathbb{E}[Z(1-D)] &= \mathbb{E}[\mathbbm{1}_{\{Z=1,D=0\}}] \\ &= \mathbb{E}[\mathbbm{1}_{\{Z=1,S=nt\}}] \\
	&=\mathbb{E}[\mathbbm{1}_{\{Z=1\}} \mathbbm{1}_{\{S=nt\}}] \\
	&= \mathbb{E}[\mathbbm{1}_{\{Z=1\}}] \mathbb{E}[\mathbbm{1}_{\{S=nt\}}] \\
	&= \pi_z \pi_{nt}.
	\end{align*} 
	Similarly,
	$\pi_{vat}=\mathbb{E}[(1-Z)D] = (1-\pi_{z})\pi_{at}$. Therefore, $1-\frac{\pi_{vnt}}{\pi_{z}}-\frac{\pi_{vat}}{1-\pi_z}=\pi_{co}$. 
	By the assumptions that $\pi_z \in (0,1)$, $\pi_{co}\neq 0$ and equation \eqref{eqLeo} we can thus conclude that 
	\begin{equation*}
		f(\boldsymbol{\beta})=\frac{\mu- \frac{\mu_{vnt}}{\pi_{z}} - \frac{\mu_{vat}}{(1-\pi_z)}}{1-\frac{\pi_{vnt}}{\pi_{z}}-\frac{\pi_{vat}}{1-\pi_z}}
		= \mu_{co},
	\end{equation*}
	where $\mu_{vnt}=\mathbb{E}[Z(1-D)X]$ and $\mu_{vat}=\mathbb{E}[(1-Z)DX]$.
	By the same assumptions, 
	\begin{equation*}
	\nabla f(\boldsymbol{\beta})=
	\renewcommand*{\arraystretch}{1.4}\small
	\begin{pmatrix} 
	1/\pi_{co}\\
	-1(\pi_{co} \pi_z)\\
	-1/(\pi_{co}(1-\pi_z))\\
	\left( \mu - \frac{\mu_{vnt}}{\pi_z}- \frac{\mu_{vat}}{(1-\pi_z)} \right) / (\pi_z \pi_{co}^2) \\ 
	\left( \mu - \frac{\mu_{vnt}}{\pi_z}- \frac{\mu_{vat}}{(1-\pi_z)} \right)  / ((1-\pi_z) \pi_{co}^2) \\ 
	\left(
	(\frac{\pi_{vat}}{(1-\pi_z)^2} - \frac{\pi_{vnt}}{\pi^2_z})\mu + 
	(1- \frac{\pi_{vat}}{(1-\pi_z)^2})\frac{\mu_{vnt}}{\pi_z^2} + 
	(\frac{\pi_{vnt}}{\pi_z^2}-1)\frac{\mu_{vat}}{(1-\pi_z)^2}
	\right) / \pi_{co}^2
	\end{pmatrix} 
	\end{equation*}
	 exists and is unequal to $\boldsymbol{0}$ and we can conclude with the delta method that
	\[
	\sqrt{N}(f(\boldsymbol{\hat{\beta}}) - \mu_{co}) \xrightarrow{d} \mathcal{N}(0,\nabla f(\boldsymbol{\beta})^T \Sigma \nabla f(\boldsymbol{\beta})).
	\]


\end{proof}

\begin{corollary}\label{corollary: avar estimation}
If in addition to the assumptions in Theorem \ref{a.dis}, $X$ has finite fourth moments, then the plug in estimator is consistent for $\hat{\mu}_{co}$'s asymptotic variance. That is,
\[
\nabla f(\boldsymbol{\hat{\beta}})^T \hat{\Sigma} \nabla f(\boldsymbol{\hat{\beta}}) \xrightarrow{P} \nabla f(\boldsymbol{\beta})^T \Sigma \nabla f(\boldsymbol{\beta}),
\]  
with $\hat{\Sigma}$ the sample covariance matrix of the $\mathbf{C}_i$.
\end{corollary}

Based on Corollary \ref{corollary: avar estimation}, we estimate asymptotically valid standard errors for our complier mean estimator $\hat{\mu}_{co}$, as 
\[
\sqrt{\frac{1}{n}\nabla f(\boldsymbol{\hat{\beta}})^T \hat{\Sigma} \nabla f(\boldsymbol{\hat{\beta}})},
\]
with $n$ the sample size.

\newpage

\section{Additional Results}
\label{results}

\begin{table}[htp]
\caption{Means and standard errors for profiling covariates}\label{tab:tab1}
\centering
\begin{tabular}{lcccc}
\toprule
Variable & Always-taker & Complier & Never-taker & Sample \\
\midrule
Proportion 0,1 &  0.21 (0.01) &  0.21 (0.01) &  0.57 (0.01) &  1.00 (0.00) \\ 
  Age (years) & 68.71 (0.57) & 64.44 (0.66) & 69.35 (0.35) & 68.17 (0.20) \\ 
  No. Comorbidities &  4.21 (0.09) &  3.31 (0.09) &  4.13 (0.05) &  3.97 (0.03) \\ 
  Sepsis 0/1 &  0.49 (0.02) &  0.23 (0.02) &  0.40 (0.01) &  0.38 (0.01) \\ 
  Female &  0.54 (0.02) &  0.61 (0.02) &  0.55 (0.01) &  0.56 (0.01) \\ 
  Hispanic &  0.06 (0.01) &  0.09 (0.01) &  0.09 (0.01) &  0.09 (0.00) \\ 
  White &  0.77 (0.02) &  0.74 (0.02) &  0.77 (0.01) &  0.76 (0.01) \\ 
  African-American &  0.12 (0.01) &  0.16 (0.01) &  0.12 (0.01) &  0.13 (0.00) \\ 
  Disability 0/1 &  0.10 (0.01) &  0.08 (0.01) &  0.10 (0.01) &  0.10 (0.00) \\ 
  Medicare &  0.66 (0.02) &  0.55 (0.02) &  0.67 (0.01) &  0.64 (0.01) \\ 
  Medicaid &  0.05 (0.01) &  0.07 (0.01) &  0.06 (0.01) &  0.06 (0.00) \\ 
  Liver Disease 0/1 &  0.30 (0.02) &  0.23 (0.02) &  0.24 (0.01) &  0.25 (0.01) \\ 
  Renal Failure 0/1 &  0.06 (0.01) &  0.05 (0.01) &  0.06 (0.01) &  0.06 (0.00) \\ 
  Lymphoma 0/1 &  0.02 (0.01) &  0.01 (0.00) &  0.02 (0.00) &  0.02 (0.00) \\ 
  Metastatic Cancer 0/1 &  0.04 (0.01) &  0.04 (0.01) &  0.04 (0.00) &  0.04 (0.00) \\ 
  Obesity 0/1 &  0.12 (0.01) &  0.11 (0.01) &  0.12 (0.01) &  0.11 (0.00) \\  
\bottomrule
\end{tabular}
\end{table}

\end{appendix}

\end{document}